\documentclass[a4paper,11pt]{article}

\usepackage{graphicx}
\usepackage[utf8]{inputenc}
\usepackage[mathscr]{euscript}
\usepackage{amsmath,amsopn,amsthm, amssymb, amsfonts, amsbsy}
\usepackage{graphicx,subfigure}
\usepackage{color}
\usepackage{enumerate}
\usepackage[english]{babel}
\usepackage{cite}
\usepackage{verbatim}
\usepackage{authblk}
\usepackage{hyperref}
\usepackage{float}
\usepackage{silence}

\newcommand{\R}{\mathbb{R}}
\newcommand{\Z}{\mathbb{Z}}

\newcommand{\B}{\mathbb{B}}
\newcommand{\Q}{\mathbb{Q}}

\newcommand{\NP}{\mathscr{NP}}
\newcommand{\Pclass}{\mathscr{P}}

\newcommand{\polytope}{\mathcal{P}}

\newcommand{\CR}{\textsc{CR}}
\newcommand{\Ccal}{\mathcal{C}}
\newcommand{\consets}{\mathcal{S}}

\DeclareMathOperator{\convexhull}{\rm{convex-hull}}

\theoremstyle{plain}
\newtheorem{theorem}{Theorem}
\newtheorem{proposition}[theorem]{Proposition}
\newtheorem{corollary}[theorem]{Corollary}
\newtheorem{lemma}[theorem]{Lemma}
\newtheorem{claim}{Claim}[theorem]

\theoremstyle{definition}
\newtheorem{problem}{Problem}

\providecommand{\keywords}[1]{\textit{Keywords:} #1}

\begin{document}

\title{Polyhedral study of the Convex Recoloring problem\protect\footnote{E-mails: \href{mailto:mcampelo@lia.ufc.br}{\texttt{mcampelo@lia.ufc.br}}, \href{mailto:phablo@dcc.ufmg.br}{\texttt{phablo@dcc.ufmg.br}}, \href{mailto:joel.cruz@sti.ufc.br}{\texttt{joel.cruz@sti.ufc.br}}
}}

\author[1]{Manoel Campêlo\thanks{Partially supported by CNPq-Brazil (Proc. 443747/2014-8, 305264/2016-8) and FUNCAP (PNE-0112-00061.01.00/16)}}
\author[2]{Phablo F. S. Moura\thanks{Supported by  grants \#2015/11937-9, \#2016/21250-3 and \#2017/22611-2, São Paulo Research Foundation (FAPESP), and CAPES.}}
\author[3]{Joel Soares}

\affil[1]{Dep. de Estatística e Matemática Aplicada, Universidade Federal do Ceará, Brazil}
\affil[2]{Dep. de Ciência da Computação, Universidade Federal de Minas Gerais, Brazil}
\affil[3]{Secretaria de Tecnologia da Informação, Universidade Federal do Ceará, Brazil}

\maketitle

\begin{abstract}
	A coloring of the vertices of a connected graph is convex if each color class induces a connected subgraph.  
	We address the convex recoloring (\CR) problem defined as follows.
	Given a graph~$G$ and a coloring of its vertices, recolor a minimum number of vertices of~$G$ so that the resulting coloring is convex.  
	This problem, known to be~$\NP$-hard even on paths, was firstly motivated by applications on perfect phylogenies.
	In this work, we study \CR\ on general graphs from a polyhedral point of view. 
	First, we introduce a full-dimensional polytope based on the idea of connected subgraphs, and present a class of valid inequalities with righthand side one that comprises all facet-defining inequalities with binary coefficients when the input graph is a tree.
	Moreover, we define a general class of inequalities with righthand side in~$\{1, \ldots, k\}$, where~$k$ is the amount of colors used in the initial coloring, and show sufficient conditions for validity and facetness of such inequalities.
	Finally, we report on computational experiments for an application on mobile networks that can be modeled by the polytope of \CR\ on paths. We evaluate the potential of the proposed inequalities to reduce the integrality gaps.

\end{abstract}%



\section{Introduction}
A \emph{total coloring} of a graph~$G$ is a function~$C$ that assigns to each vertex
in~$V(G)$ a color, that is, $C \colon V(G) \to \Ccal$, where~$\Ccal$ is a set of colors.
A \emph{partial coloring} (or simply, a \emph{coloring}) of~$G$ is a
function~$C \colon V(G) \to \Ccal \cup \{\emptyset\}$, where~$\emptyset$ indicates absence of
color.  We say that a vertex~$v$ in~$G$ is \emph{uncolored} if~$C(v)= \emptyset$, and we
say that~$G$ is \emph{$k$-colored} if its coloring uses $k$~colors.  For each~$c \in \Ccal$, the \emph{color class}~$c$ (denoted by~$C^{-1}(c)$) is the set of
vertices in~$G$ that have color~$c$, that is, $C^{-1}(c):=\{ v \in V(G) \colon C(v)=c\}$.
Note that the coloring defined here differs from the classic vertex (proper) coloring, in which
adjacent vertices have different colors.

A \emph{totally colored graph} is a pair~$(G,C)$ consisting of a graph~$G$ and a total coloring $C$ of its vertices.  
A total coloring~$C$ is \emph{convex} if the color class~$c$ induces a connected subgraph of~$G$  for all~$c \in \Ccal$.  Analogously, we define a (partially) \emph{colored graph}.
A \emph{convex coloring}~$C \colon V(G) \to \Ccal \cup \{\emptyset\}$ is a coloring that can be extended to a convex total coloring by solely assigning a color in~$\Ccal$ to each uncolored vertex.
Given a colored graph~$(G,C)$, any other coloring~$C^\prime \colon V(G) \to \Ccal\cup\{\emptyset\}$ is a \emph{recoloring}
of~$(G,C)$, where~$\Ccal$ is the same set of colors of the initial coloring~$C$.  
A vertex~$v$ is said to be \emph{recolored} (by~$C^\prime$) if~$C(v)\neq \emptyset$ and~$C(v) \neq C^\prime(v)$.
We address the following problem.

\begin{problem}{\textsc{Convex Recoloring} $(\CR)$  \\}
      \textsc{Instance:} A connected graph~$G$, a coloring~$C\colon V(G) \to \Ccal\cup \{\emptyset\}$ and a weight function $w \colon V(G) \to \Q_\geq$.\\
      \textsc{Find:} A convex recoloring~$C^\prime$ of~$(G,C)$.\\
      \textsc{Goal:} Minimize~$\sum_{v \in R_C(C^\prime)} w(v)$, where~$R_C(C^\prime):= \{ v \in V(G) \colon C(v)\neq \emptyset$ 
      and~$ C(v)\neq C^\prime(v) \}$ is the set of vertices recolored by~$C^\prime$.
\end{problem}

The $\CR$ problem has been intensively investigated under different approaches: approximation algorithms, inapproximability, integer linear programming and heuristics.
We next mention some results for totally colored graphs.  
Kanj and Kratsch~\cite{KanKra09} proved that the problem is $\NP$-hard for paths even if each color appears at most twice.
Camp\^elo~et~al.~\cite{CamHuiSamWak14} showed that the $\CR$ problem is $\NP$-hard on unweighted $2$-colored grids.  Approximation algorithms for the unweighted case have also been designed: with ratio~$(2+\varepsilon)$ for bounded
treewidth graphs~\cite{KamTho12}, and ratio~$2$ for paths~\cite{MorSni07}.  
The first constant factor  approximation algorithm  for \CR\ on general graphs was recently proposed by Bar-Yehuda et al.~\cite{BarKutRaw18}. 
They designed a $3/2$-approximation for graphs in which each color appears at most twice.
On the other hand, Moura and Wakabayashi~\cite{MouWak2019} showed that unweighted $\CR$ cannot be approximated within a ratio of $n^{1-\varepsilon}$ (for any $\varepsilon > 0$) on $n$-vertex $k$-colored bipartite graphs, unless $\Pclass=\NP$.

An integer linear formulation for~$\CR$ on general graphs was introduced by Camp\^elo~et~al. in~\cite{CamFreLimMouWak16}. 
They presented a detailed polyhedral study and computational experiments on trees.
A compact formulation  for $\CR$ on trees was proposed by Chopra~et~al.~in~\cite{ChoFilLeeRyuShiVan16}. 
They showed that their ILP formulation dominates the one described in~\cite{CamFreLimMouWak16}.

In 2017, Moura~\cite{Mou17} (the second author of this paper) introduced an ILP formulation based on connected subgraphs for \CR\ on general graphs and showed preliminary computational experiments of a column-generation algorithm on trees.
Independently, Chopra~et~al.~\cite{ChoErdKimShi17} also devised a column-generation approach to \CR\ restricted to trees. 
We remark that the ILP formulation in~\cite{ChoErdKimShi17} is a particular case of the model proposed in~\cite{Mou17},
and that no polyhedral study is presented in those works.
The computational experiments in~\cite{ChoErdKimShi17}  and~\cite{Mou17} indicate that a column-generation approach for \CR\ on trees is very efficient if the number of colors is large when compared with the size of the tree.

Recently, a GRASP heuristic for \CR\ on general graphs was devised by Dantas et al.~\cite{DanSouDia2019}.
The authors compare their heuristic algorithm with a Branch-and-Bound algorithm based on the ILP formulation due to Camp\^elo~et~al.~\cite{CamFreLimMouWak16}.


In Section~\ref{sec:polytope}, we introduce a polytope based on the idea of connected subgraphs and show some properties of its facet-defining inequalities.
In Section~\ref{sec:bin-ineq}, we present a class of valid inequalities with righthand side one and show it
comprises all facet-defining inequalities with binary coefficients when the input graph is a tree.
We define a general class of inequalities with righthand side in~$\{1, \ldots, k\}$, where~$k$ is the amount of colors used in the coloring, in Section~\ref{sec:general-ineq}.
Furthermore, we present sufficient conditions for validity and facetness of such inequalities.
The potential of these inequalities to reduce the integrality gaps is evaluated in Section~\ref{sec:experiments}, particularly in the context of an application in 4G mobile networks.
\section{The connected subgraph polytope}
\label{sec:polytope}

We denote by $(G,C,w)$ an instance of~$\CR$ consisting of a connected graph~$G$, a partial coloring $C \colon V(G)  \to
\Ccal \cup \{\emptyset\}$  and  a weight function $w \colon V(G) \to \Q_\geq$. 
Let~$V_\emptyset$ be the set of uncolored vertices of~$(G,C)$.
Henceforth, we assume that the weight of each vertex in~$V_\emptyset$ is~$0$. 
Let $n:=|V(G)|\geq 3$ and $k:=|\Ccal|\geq 2$.  

Note that minimizing the sum of the weights of recolored vertices is equivalent to maximizing
the sum of the weights of vertices that keep their initial colors.
Indeed, if~$C^*$ be a feasible solution to instance  $(G,C,w)$ of $\CR$,
and $R(C^*)= \{ v \in V(G) \colon C(v)\neq \emptyset$ and~$ C(v)\neq C^*(v) \}$ is the the set of vertices recolored by~$C^*$,
then one may easily verify that 
\(\sum_{v \in R(C^*)} w(v) = \sum_{v \in V(G)} w(v) - \sum_{v \in V(G) \setminus R(C^*) } w(v).\)
This remark is important since we shall define our integer linear programming (ILP) formulation for $\CR$ as a maximization problem.


In what follows, we propose an ILP formulation for the convex recoloring problem on general graphs.
Before presenting our formulation, let us introduce some notation.
For each~$v \in V(G)$ and~$c \in \Ccal$, we define a constant~$w_{v,c}$ which is $w(v)$ if~$C(v) = c$, and~$0$ otherwise.
Let $\consets(G)$ be the collection of all subsets of~$V(G)$ that induce connected subgraphs of~$G$.
We denote by~$\eta$ the cardinality of~$\consets(G)$.
We prefer the simplified notation~$\consets$ when~$G$ is clear from the context.
For each $H \in \consets(G)$ and  $c \in \Ccal$,  we say that $H$ is $c$-\emph{monochromatic} if color~$c$ is assigned to all vertices in~$H$, 
and we define $w_{H, c} = \sum_{v \in V(H)} w_{v,c}$ as the weight of making~$H$ $c$-monochromatic, that is, the total weight given by~$H$ if we assign color~$c$ to all its vertices.

We define, for each $H \in \consets(G)$ and  $c \in \Ccal$, a binary variable $x_{H,c}$ with the following meaning: 
$x_{H,c} = 1$  if $H$ is $c$-monochromatic.
Let~$\consets_\cap(v) = \{H \in \consets(G) \colon v \in H\}$.
We introduce the following formulation for the Convex Recoloring problem.

\begin{align} 
 	\!\!\!  \!\!\!&& \max \:&  \sum_{H  \in \consets(G)} \sum_{c \in \Ccal}  w_{H,c} \:  x_{H,c} \label{fob} \\
 	&& \text{s.t.} \:& \sum_{H \in \consets_{\cap}(v) } \sum_{c \in \Ccal} x_{H,c}  \leq  1, && \!\!\!\text{ for all } v \in V(G) \label{ineq:vertices}\\ 
 	&&	& \sum_{H \in \consets(G)}  x_{H,c}  \leq  1, && \!\!\!\text{ for all } c \in \Ccal \label{ineq:colors} \\
 	&&	& x_{H,c}  \in  \B, && \!\!\!\text{ for all } H \in \consets(G), \:c \in \Ccal \label{ineq:integer}
\end{align}

Inequalities~\eqref{ineq:vertices} require that each vertex receives at most one color, 
while inequalities~\eqref{ineq:colors}  require that each color is assigned to at most one connected subgraph.
This formulation was introduced in~\cite{Mou17} and, independently, in~\cite{ChoErdKimShi17} for \CR\ restricted to trees.
We next show that every vector that satisfies constraints~\eqref{ineq:vertices}--\eqref{ineq:integer} corresponds to a convex coloring of~$G$, and vice-versa. 

Given a convex (partial) coloring~$C \colon V(G) \to \Ccal$, we define the vector~$\chi(C) \in \B^{\eta k}$ as follows.
For every~$H \in \consets(G)$ and $c \in \Ccal$,  $\chi(C)_{H,c}=1$ iff $H = G[C^{-1}(c)]$.
Furthermore, let us define \[\polytope(G, C) = \convexhull(\{ x \in \R^{\eta k} \colon x \text{ satisfies } \eqref{ineq:vertices}, \eqref{ineq:colors} \text{ and } \eqref{ineq:integer}\}).\]


\begin{proposition}\label{prop:correctness}
 $\polytope(G, C) =  \convexhull(\{\chi(C^\prime)  \in \B^{\eta k} \colon C^\prime \text{ is a convex col. of } G\})$.
\end{proposition}
\begin{proof}
  Let~$C^\prime$ be a convex coloring of~$G$.
  One may easily check that $\chi(C^\prime)$ satisfies inequalities~\eqref{ineq:vertices}, \eqref{ineq:colors} and~\eqref{ineq:integer}.
  Therefore, we conclude that~$\chi(C^\prime)$ belongs to~$\polytope(G, C)$.
  
  Consider now an integer point~$x \in \polytope(G, C)$ and let~$\consets^\prime = \{H \in \consets(G) \colon x_{H,c} = 1 \text{ for some } c \in \Ccal\}$.
  We define~$C^\prime \colon V(G) \to \Ccal$ in the following way.
  For every~$H \in \consets(G)$ and~$c \in \Ccal$ such that~$x_{H,c} = 1$,  we set $C^\prime (v)=c$ for all~$v \in H$.
  For all~$v \in V(G) \setminus (\bigcup_{H \in \consets^\prime} H)$, we define~$C^\prime (v) = \emptyset$.
  
  Since~$x$ satisfies inequalities~\eqref{ineq:vertices}, it follows that~$x_{H, c} + x_{H^\prime, d} \leq 1$ for all~$H, H^\prime \in \consets(G)$ such that $H \cap H^\prime \neq \emptyset$, and for all colors~$c,d \in \Ccal$.
  Hence, the elements of~$\consets^\prime$ are pairwise disjoint.
  Because $x$ satisfies inequalities~\eqref{ineq:colors}, every color in~$\Ccal$ is associated with at most one element of~$\consets^\prime$.
  Using these remarks and by the fact that every set~$H \in \consets^\prime$ induces a connected subgraph of~$G$, we conclude that~$C^\prime$ is a convex coloring of~$G$. 
\end{proof}

Let~$H \in \consets(G)$ and~$c \in \Ccal$.
We define~$e(H,c) \in \B^{\eta k}$ to be the unit vector such that it single non-null entry is~$e(H,c)_{H,c} = 1$.
Since~$H$ induces a connected subgraph of~$G$, it trivially holds that $e(H,c)$ belongs to~$\polytope(G,C)$.

\begin{proposition}\label{prop:dimension}
 The following hold:
 \begin{enumerate}[(i)]
  \item $\polytope(G, C)$ is full-dimensional, that is, $\dim (\polytope(G, C)) = \eta k$.
  \item For every~$ H \in \consets(G)$ and~$c \in \mathcal{C}$, $x_{H, c} \geq 0$ is facet-defining.
 \end{enumerate}
\end{proposition}
\begin{proof}
Let~$X$ be the set of vectors~$\{e(H,c) \in \B^{\eta k} \colon H \in \consets(G) \text{ and } c \in \Ccal\}$.
Note that~$|X| = \eta k$, the vectors in~$X$ are linearly independent and that~$X \cup \{\vec 0\}\subseteq \polytope_k(G, C)$.
As a consequence, we have~$\dim(\polytope_k(G, C))= \eta k$.

Consider~$H \in \consets(G)$ and~$c \in \Ccal$. 
The vectors in~$(X \setminus\{e(H,c)\} )\cup \{\vec 0\}$ are affinely independent and belong to the face~$\{x \in \polytope_k(G, C) \colon x_{H,c}=0\}$.
Therefore, $x_{H,c}\geq 0$ induces a facet of~$\polytope_k(G, C)$.
\end{proof}


Actually, we next prove that $-x_{H,c}\leq 0$ are the unique $\leq$-inequalities with a negative coefficient that define facets of~$\polytope(G, C)$. 
Moreover, in the other facet-defining inequalities, the coefficient of $x_{G,c}$ is equal to the righthand side.

\begin{lemma}\label{lemma:facet-negative-coef}
  Let $(\pi,\pi_0)$ be a valid inequality on integer coefficients that is facet-defining for $\polytope(G, C)$.
  The following hold:
  \begin{enumerate}[(i)]
    \item\label{lemma:claim1}  If $\pi\neq -e(H,c)$ for all $H\in \consets(G)$ and $c \in \Ccal$, then $\pi \geq \vec 0$ and $\pi_0\geq 1$. 
    \item\label{lemma:claim2}  If $\pi_0\geq 1$, then $\pi\geq \vec 0$ and $\pi_{H,c}\leq \pi_{G, c}=\pi_0$ for all $H \in \consets(G)$ and $c\in \Ccal$. 
  \end{enumerate}
\end{lemma}
\begin{proof}
 Let~$F$ be the facet of~$\polytope(G, C)$ induced by inequality~$(\pi, \pi_0)$, that is, $F=\{ x \in \polytope(G, C) \colon \pi x = \pi_0\}$.
 To prove~\eqref{lemma:claim1}, suppose that~$F$ is different from~$F_{H,c} = \{x \in \polytope(G, C) \colon x_{H,c}=0\}$, and thus there is a vector~$\bar x \in F$ such that~$\bar x_{H,c} = 1$.
 Let~$\hat x$ be the vector equals~$\bar x$ except for~$\bar x_{H,c}=0$.
 Note that~$\hat x$ also belongs to~$\polytope(G, C)$.
 Thus, we have~$\pi \hat x \leq \pi_0 = \pi \bar x$ which implies~$\pi_{H,c} \geq 0$.
 Therefore, we conclude that~$\pi \geq \vec 0$ and~$\pi_0\geq 0$.
 
 Suppose to the contrary that~$\pi_0=0$.
 Since~$\pi \neq \vec 0$, there exists~$H \in \consets(G)$ and~$c \in \Ccal$ such that~$\pi_{H,c} > 0$.
 Thus, every vector~$x$ in~$F$ satisfies~$x_{H,c} = 0$, that is, $F \subset F_{H,c}$, a contradiction.
 Consequently, it follows that~$\pi_0\geq 1$.
 This completes the proof of Claim~\eqref{lemma:claim1}.
 
 To prove Claim~\eqref{lemma:claim2}, suppose~$\pi_0 \geq 1$.
 Let~$c \in \Ccal$ and let~$\pi^\prime$ be equal to~$\pi$ except for~$\pi^\prime_{G, c} = \pi_0$.
 Consider an integer point~$x \in~\polytope(G, C)$.
 If $x_{G, c}=0$, then $\pi^\prime x= \pi x \leq \pi_0$. 
 Otherwise, $x_{G,c}=1$ and, consequently, all other entries of~$x$ are null since~$V(G)$ clearly intersects all subgraphs.
 Thus, we obtain~$\pi^\prime x = \pi_0$.
 Therefore, $(\pi^\prime, \pi_0)$ is a valid inequality for~$\polytope(G, C)$.
 
 Since~$(\pi^\prime,\pi_0)$ cannot dominate the facet-defining inequality~$(\pi,\pi_0)$, it must hold that $\pi^\prime_{V(G), c} \geq \pi_0$.
 On the other hand, as $e(H,c) \in \polytope(G, C)$, we have $\pi_{H, c} \leq \pi_0$ for all~$H \in \consets(G)$ and~$c \in \Ccal$.
 By Claim~\eqref{lemma:claim1} of this lemma, it follows that $\pi\geq \vec 0$. 
\end{proof}

\begin{corollary}\label{cor:equivalences}
Let~$(\pi,\pi_0)$ be a facet-defining inequality of~$\polytope(G,C)$ such that~$\pi \in \Z^{\eta k}$. 
The following statements are equivalent: 
\begin{enumerate}[(i)]
	\item $\pi_0=1$.
	\item $\pi_{V(G),c}=1$ for all  $c \in \Ccal$.
	\item $\pi\in \B^{\eta k}$ with $\pi\neq \vec 0$.
\end{enumerate}
\end{corollary}

\section{Facets with binary coefficients} \label{sec:bin-ineq}

We present a class of valid inequalities with righthand side one and show that it comprises all facet-defining inequalities with binary coefficients when the input graph is a tree.

For $H \in \consets(G)$, let us define $\consets_{\cap}(H)=\{H^\prime \in \consets(G) \colon  H^\prime \cap H \neq \emptyset \}$ and  $\consets_{\supset}(H)=\{H^\prime \in \consets(G) \colon H^\prime \supseteq H \}$
as the collections of  connected subgraphs of~$G$ that  intersect and contain~$H$, respectively.
It is clear that~$\consets_{\supset} (H) \subseteq \consets_{\cap} (H)$.

\begin{theorem}\label{theorem:facet1}
For every $H \in \consets(G)$ and $c \in \Ccal$, the inequality 
\begin{equation} \label{eq:ineq1}
   \sum_{H^\prime \in \consets_{\supset}(H)} \sum_{c^\prime \in \Ccal \setminus \{c\}} x_{H^\prime, c^\prime} + \sum_{H^\prime \in \consets_{\cap}(H)} x_{H^\prime, c} \leq 1
\end{equation} 
is valid for $\polytope(G, C)$.
Moreover, it is facet-defining if $|\Ccal| \geq 2$.
\end{theorem}
\begin{proof}
Let~$x$ be an extreme point of~$\polytope(G, C)$.
Let $A = \sum_{H^\prime \in \consets_{\supset}(H)} \sum_{c^\prime \in \Ccal \setminus \{c\}} x_{H^\prime, c^\prime}$ and let $B=\sum_{H^\prime \in \consets_{\cap}(H)} x_{H^\prime, c}$.
Because of inequalities~\eqref{ineq:vertices}, it holds that~$A\leq 1$.
Inequality~\eqref{ineq:colors} for color~$c$ guarantees that~$B\leq1$.
Observe that, if~$B=1$, then there exists a vertex~$v$ which is assigned color~$c$.
Moreover, $v$ is a vertex of every set~$H^\prime \in \consets_\supset (H)$.
Thus, we have~$A=0$ due to constraints~\eqref{ineq:vertices} of the formulation.
This proves that inequality~\eqref{eq:ineq1} is valid for~$\polytope(G, C)$.

Let $\pi x\leq \pi_0$ be inequality~\eqref{eq:ineq1} and~$F_{H, c} = \{ x \in \polytope(G,C) \colon \pi x = \pi_0\}$ be the corresponding face. 
We shall prove that $F_{H, c}$  is a facet by showing that,  for every valid inequality~$(\lambda, \lambda_0)$ of~$\polytope(G,C)$, we have~$(\lambda, \lambda_0) = \lambda_0(\pi,\pi_0)$  if~$F_{H, c} \subseteq \{x \in \polytope(G,C) \colon \lambda x = \lambda_0 \}$. 
Let~$H^\prime \in \consets(G)$ and let~$c^\prime$ be a color in~$\Ccal \setminus\{c\}$, which exists because $|\Ccal| \geq 2$.
          
First, we show the coefficients associated with~$c$.
If $H^\prime \in \consets_{\cap}(H)$, then $e(H^\prime, c) \in F_{H, c}$ implies~$\lambda_{H^\prime, c} = \lambda_0$.
If $H^\prime \notin \consets_{\cap}(H)$,  then the points~$e(H, c^\prime)$ and~$e(H, c^\prime)+e(H^\prime, c)$ both belong to~$F_{H, c}$.
Thus, we have~$\lambda (e(H, c^\prime)+e(H^\prime, c))=\lambda e(H, c^\prime)$, that is, $\lambda_{H^\prime, c}=0$. 

If $H^\prime \in \consets_{\supset}(H)$, then $e(H, c)\in F_{H,c}$ and $e(H^\prime, c^\prime) \in F_{H, c}$ show that~$\lambda_{H,c}=\lambda_{H^\prime, c^\prime}$. 
If~$H^\prime \notin \consets_{\supset}(H)$, then there exists a connected graph $\hat H$ which is subgraph of $H - H^\prime$. 
Observe that $\hat H \in \consets_{\cap}(H)$. 
Therefore, $e(\hat H, c)$ and $e(\hat H, c)+e(H^\prime, c^\prime)$, both in~$F_{H,c}$, lead to $\lambda_{H^\prime, c^\prime}=0$.
\end{proof}

It is worth noting that inequalities~\eqref{eq:ineq1} dominate and generalize the formulation constraints. 
Indeed, inequalities~\eqref{ineq:colors} are dominated by \eqref{eq:ineq1} for $H=G$ whereas~\eqref{ineq:vertices} are equivalent to~\eqref{eq:ineq1} for~$H$ such that~$|V(H)|=1$.

We next show that the class of inequalities~\eqref{eq:ineq1} comprises all facet-defining inequalities on integer coefficients with righthand side one if~$G$ is a tree.

\begin{theorem}\label{thm:facet}
Let $k\geq 2$ and let~$\pi\in \Z^{\eta k}$.
If~$(\pi, 1)$ is a facet-defining inequality for~$\polytope(G,C)$ and~$G$ is a tree, then~$(\pi, 1)$ is \eqref{eq:ineq1} for some~$H^* \in \consets(G)$ and~$c^* \in \Ccal$.
\end{theorem}
\begin{proof}
First note that~$\pi\in \B^{\eta k}$ because of Corollary~\ref{cor:equivalences}. 
Thus, the inequality has the form~$\sum_{c \in \Ccal}  \sum_{H \in \consets_\pi(c)} x_{H,c} \leq 1$, where $\consets_\pi(c)=\{ H \in \consets(G) \colon \pi_{H,c}=1\}$. 
Therefore, it suffices to prove that there exist $H^*\in \consets(G)$ and $c^*\in \Ccal$ such that
\begin{enumerate}[(i)]
	\item $\consets_\pi(c^*) = \consets_{\cap}(H^*)$, and \label{thm:property1}
	\item $\consets_\pi(c) = \consets_{\supset} (H^*)$ for all $c \in \Ccal\setminus\{c^*\}$. \label{thm:property2}
\end{enumerate}

\begin{claim}
  Let $c, c^\prime \in \Ccal$ such that $\bigcap_{H \in \consets_\pi(c)} H \neq \emptyset$ and $\bigcap_{H \in \consets_\pi(c^\prime)} H \neq \emptyset$. It holds that~$\consets_\pi(c)=\consets_\pi(c^\prime)$.
\end{claim}
\label{claim:1}
\begin{proof}
  Since~$c$ and~$c^\prime$ play indistinct roles, it is enough to prove that $\consets_\pi(c)\subseteq \consets_\pi(c^\prime)$ for $c\neq c^\prime$. 
  Suppose to the contrary that there exists~$H \in \consets_\pi(c)\setminus \consets_\pi(c^\prime)$, that is, $\pi_{H,c}=1$ and~$\pi_{H, c^\prime}=0$. 
  Let~$\pi^\prime$ differ from~$\pi$ only by $\pi^\prime_{H,c^\prime}=1$. 
  We next show that $(\pi^\prime,1)$ is valid for~$\polytope_k(G,C)$. 
  
  Let $x$ be an extreme point of $\polytope(G,C)$.
If~$x_{H,c^\prime}=0$, then $\pi' x=\pi x\leq 1$. 
Suppose now that~$x_{H,c^\prime}=1$ and so~$x_{H,c}=0$. 
Let~$x^\prime$ be identical to~$x$ except for $x^\prime_{H,c^\prime}=0$ and $x^\prime_{H,c}=1$.
For every~$H^\prime \neq H$ such that~$x_{H^\prime, c}=1$, we also set~$x^\prime_{H^\prime, c}=0$.
Note that such~$H^\prime$ satisfies~$H^\prime \notin \consets_{\pi}(c)$, that is, $\pi_{H^\prime, c}=0$, otherwise, by the hypothesis, we have~$H\cap H^\prime \neq \emptyset$, a contradiction to the fact that~$x$ is a feasible point with~$x_{H, c^\prime} = x_{H^\prime, c}=1$.
Observe now that $x^\prime \in \polytope_k(G, C)$ and so~$\pi^\prime x=\pi x^\prime\leq 1$. 
Thus, $(\pi^\prime, 1)$ is valid and dominates~$(\pi, 1)$, a contradiction. 
Therefore, $\consets_{\pi}(c)\subseteq \consets_{\pi}(c^\prime)$ and so the claim holds. 
\end{proof}

\begin{claim}
  Let $c \in \Ccal$ such that  $\bigcap_{H \in \consets_{\pi}(c)} H = \emptyset$. 
For every~$c^\prime \in \Ccal \setminus \{c\}$,  we have~$\bigcap_{ H \in \consets_{\pi}(c^\prime)} H \neq \emptyset$.
\end{claim}
\label{claim:2}
\begin{proof}
By Lemma~\ref{lemma:facet-negative-coef}, $\consets_{\pi}(c)$ is nonempty and so there must exist~$H_1$ and~$H_2$ in~$\consets_{\pi}(c)$ such that $H_1\cap H_2=\emptyset$.
Let~$v_1$ and~$v_2$ be the two closest vertices in~$G$ such that~$v_1 \in H_1$ and~$v_2 \in H_2$.
Since~$G$ is a tree, every path from~$H_1$ to~$H_2$ must contain both~$v_1$ and~$v_2$
Let $c^\prime \in \Ccal \setminus \{c\}$ and let~$H^\prime \in \consets_{\pi}(c^\prime)$ (note that $\consets_{\pi}(c^\prime)\neq \emptyset$ by Lemma~\ref{lemma:facet-negative-coef}).
We next prove that~$H^\prime \in \consets_\cap(H_1)\cap \consets_\cap(H_2)$.
As a consequence, we have~$H^\prime \supseteq \{v_1, v_2\}$ which  implies~$\bigcap_{H\in \consets_{\pi}(c^\prime)} H \neq \emptyset$.

Suppose to the contrary that~$ H^\prime \notin \consets_\cap(H_1)$. 
Thus, $x=e(H_1, c)+e(H^\prime, c^\prime) \in \polytope_k(G,C)$ and~$\pi x=2$. 
This contradicts the validity of~$(\pi, 1)$. 
Therefore, $H^\prime$ belongs to~$\consets_\cap(H_1)$. 
Similarly, we prove that~$H^\prime$ belongs~$\consets_\cap(H_2)$. 
\end{proof}

Since $k \geq 2$, by Claim~6.2
, there is~$c^*\in \Ccal$ such that $\bigcap_{H \in \consets_{\pi}(c)} H \neq \emptyset$ for all~$c \in \Ccal\setminus \{c^*\}$. 
By Claim~6.1
, $\consets_{\pi}(c)=\consets_{\pi}(c^\prime)$ for all~$c,c^\prime \in \Ccal\setminus\{c^*\}$. 
Therefore, there is a connected subgraph~$H^*$ (i.e. a subtree) such that
  \begin{equation}\label{thm:equation}
    \bigcap_{H \in \consets_{\pi}(c)} H=H^*, \forall c \in \Ccal \setminus \{c^*\}
  \end{equation}
We shall prove that~$H^*$ and~$c^*$, as defined above, satisfy  properties~\eqref{thm:property1} and~\eqref{thm:property2}.
  
\begin{claim}
  Let $c\in \Ccal$ and $H \in \consets_{\pi}(c)$. 
  It holds that~$\consets_{\supset}(H) \subseteq \consets_{\pi}(c)$.
\end{claim}
\label{claim:3}
\begin{proof}
Suppose to the contrary that there is~$H^\prime \in \consets_{\supset}(H) \setminus \consets_{\pi}(c)$.
Of course, it holds that~$H^\prime \neq H$.
Let us define~$\pi^\prime$ from~$\pi$ by changing $\pi^\prime_{H^\prime, c}$ from~0 to~1.
We now show that~$\pi^\prime x\leq 1$ is valid for~$\polytope_k(G,C)$. 

Let~$x$ be an integer point of $\polytope_k(G,C)$.
If $x_{H^\prime, c}=0$, then $\pi^\prime x=\pi x\leq 1$.
Otherwise, $x_{H^\prime, c}=1$ and, consequently, $x_{H,c}=0$.
We define~$x^\prime$ from~$x$ by setting~$x^\prime_{H^\prime, c}=0$ and~$x^\prime_{H,c}=1$. 
Observe that~$x^\prime$ belongs to~$\polytope_k(G,C)$ since $H^\prime \supseteq H$.
Moreover, we have~$\pi^\prime x=\pi x^\prime\leq 1$ since $H \in \consets_{\pi}(c)$.
Hence, $(\pi^\prime,1)$ is valid for~$\polytope_k(G,C)$ and it dominates~$(\pi,1)$, a contradiction. 
\end{proof}

We now prove property~\eqref{thm:property1}.
Suppose to the contrary that there is $H \in \consets_{\pi}(c^*)\setminus \consets_{\cap}(H^*)$.
Consider a color~$ c \neq c^*$. 
By~\eqref{thm:equation}, there must exist $H^\prime \in \consets_{\pi}(c)$ such that $H^\prime \notin \consets_{\cap}(H)$.
Otherwise, if every subtree in~$\consets_{\pi}(c)$ intersected~$H$, then~$H$ would intersect~$H^*$. 
Thus, $x=e(H^\prime, c)+e(H, c^*) \in \polytope_k(G,C)$ and $\pi x=2$, which contradicts the valid of~$(\pi,1)$. 
Therefore, $\consets_{\pi}(c^*)\subseteq \consets_{\cap}(H^*)$.
  
To prove that $\consets_{\cap}(H^*) \subseteq \consets_{\pi}(c^*)$,  note that~$\consets_{\cap} (H^*) \subseteq \bigcup_{v \in H^* } \consets_{\supset} (v)$.
By Claim~6.3
, it suffices to prove that the trivial graphs in $\consets_{\cap}(H^*)$  (i.e. all vertices in~$H^*$) belong to $\consets_{\pi}(c^*)$.
Consider a vertex $v \in H^*$.
Suppose to the contrary that $v \notin \consets_{\pi}(c^*)$. 
We define~$\pi^\prime$ from~$\pi$ by setting~$\pi^\prime_{v,c^*}$ to one.
We now show that~$(\pi^\prime, 1)$ is valid for~$\polytope_k(G,C)$.
Let~$x$ be an integer point of~$\polytope_k(G,C)$.
If~$x_{v, c^*}=0$, then $\pi^\prime x=\pi x\leq 1$. 
Otherwise, $x_{v, c^*}=1$ and so~$x_{H^\prime,  c^*}=0$ for every set~$H^\prime \neq \{v\}$.
Moreover, for every~$c \in \Ccal \setminus \{c^*\}$ and~$H^\prime \in \consets_{\pi}(c)$, it follows from equation~\eqref{thm:equation} that $H^\prime \in \consets_\supset(H^*) \subseteq \consets_\supset(v)$ which implies~$x_{H^\prime, c}=0$.
Hence, $\pi x=\sum_{H^\prime \in \consets_{\pi}(c^*)} x_{H^\prime, c^*} + \sum_{c \in \Ccal\setminus \{c^*\}}\sum_{H^\prime \in \consets_{\pi}(c)}x_{H^\prime, c}= \pi_{v,c^*}=0$ and $\pi^\prime x=\pi^\prime_{v,c^*}=1$. 
Therefore, $(\pi^\prime, 1)$ is valid for~$\polytope_k(G,C)$ and dominates~$(\pi,1)$, a contradiction.
This concludes the proof of~\eqref{thm:property1}.

Consider a color~$c \in \Ccal\setminus \{c^*\}.$
Because of~\eqref{thm:equation}, $H \in \consets_{\pi}(c)$ directly implies~$H\in \consets_{\supset}(H^*)$. 
To prove that~$\consets_{\supset}(H^*) \subseteq \consets_{\pi}(c)$, by Claim~6.3
, it suffices to show that~$H^*\in \consets_{\pi}(c)$. 
Suppose to the contrary that~$\pi_{H^*,c}=0$. 
Let~$\pi^\prime$ be equal to $\pi$ except for~$\pi^\prime_{H^*,c}=1$. 
We next prove that~$(\pi^\prime, 1)$ is valid for~$\polytope_k(G,C)$.
Let~$x$ be an integer point of~$\polytope_k(G,C)$.
If~$x_{H^*,c}=0$, then~$\pi^\prime x=\pi x\leq 1$. 
Otherwise, it holds that~$x_{H^*,c}=1$ and, consequently by item~\eqref{thm:property1}, $x_{H,c^*}=0$ for all~$H \in \consets_{\pi}(c^*)$.
Let~$x^\prime$ be obtained from~$x$ by setting~$x^\prime_{H^*,c}=0$, $x^\prime_{H^*,c^*}=1$ and~$x^\prime_{H, c^*}=0$ for every~$H \neq H^*$. 
Since~$H^*\in \consets_\cap(H^*)=\consets_\pi(c^*)$, we have
\[\pi^\prime x - \pi x^\prime= (\pi^\prime_{H^*,c} x_{H^*,c}- \pi_{H^*,c}x^\prime_{H^*,c}) + (\pi^\prime_{H^*, c^*}x_{H^*, c^*}-\pi_{H^*, c^*} x^\prime_{H^*, c^*})=1-1=0.\]
As~$x^\prime \in \polytope_k(G,C)$, it follows that~$\pi^\prime x=\pi x^\prime \leq 1$. 
Therefore, $(\pi^\prime, 1)$ is valid for~$\polytope_k(G,C)$ and dominates~$(\pi,1)$, a contradiction.
\end{proof}
\section{Generalized inequalities}\label{sec:general-ineq}

Let~$\Ccal^\prime \subseteq \Ccal$ be a nonempty set of colors and let~$H \in \consets(G)$.
We now introduce the inequality
\begin{equation}\label{eq:general}
A(x) + B(x) \leq |\Ccal^\prime|,
\end{equation}
where $\delta_{H^\prime}=|\Ccal^\prime|-|H\setminus H^\prime|$ for all $H^\prime \in \consets(G)$,
\begin{align*}
A(x)=\sum_{c \in \Ccal^\prime}\sum_{H^\prime \in \consets_{\cap}(H)} \!\!\! \max(\delta_{H^\prime},1) x_{H^\prime,c} & \quad \text{ and } \\
B(x)=\sum_{c \in \Ccal \setminus \Ccal^\prime} \sum_{H^\prime \in \consets_{\cap}(H)}\!\!\! \max(\delta_{H^\prime},0) x_{H^\prime,c}. &
\end{align*}

Observe that inequality~\eqref{eq:ineq1} for~$i\in \Ccal$ is exactly~\eqref{eq:general} for~$\Ccal^\prime=\{i\}$, that is, this new class of inequalities generalizes that one from Section~\ref{sec:bin-ineq}. 
Indeed, if $|\Ccal^\prime|=1$, then~$\delta_{H^\prime}=1$ for all $H^\prime \in \consets_\supset(H)$, and~$\delta_{H^\prime}\leq 0$ for all $H^\prime \in \consets_\cap(H) \setminus \consets_\supset(H)$.

\begin{proposition} \label{prop:validade6}
If $|\Ccal^\prime| \leq |H|$ or $|H|=1$, then inequality~\eqref{eq:general} is valid for $\polytope_k(G,C)$.
\end{proposition}
\begin{proof}
Let us define~$f(x)=A(x)+B(x)$ for every~$x \in \polytope_k(G,C)$.
Consider an integer point~$x$ of~$\polytope_k(G,C)$.
Observe that, if there is~$H^\prime \in \consets(G)\setminus \consets_\cap(H)$ such that~$x_{H^\prime, c}=1$ for some~$c \in \Ccal$, then the vector~$x^\prime$ which is equal to~$x$ except for~$x^\prime_{H^\prime,c}=0$ belongs to~$\polytope_k(G,C)$ and satisfies~$f(x^\prime)=f(x)$.
Hence, we assume without loss of generality that every connected subgraph used by~$x$ intersects~$H$, that is, $\{H^\prime \colon x_{H^\prime,c}=1 \text{ for some } c\in \Ccal\} \subseteq \consets_\cap(H)$. 

First, we shall prove that there exists an integer point~$x^\prime \in \polytope_k(G,C)$ such that~$x^\prime$ only assigns colors of~$\Ccal^\prime$ and~$f(x^\prime)=f(x)$.
Let~$K$ be the subset of colors of~$\Ccal\setminus \Ccal^\prime$ that are assigned by~$x$, that is, $K=\{c \in \Ccal\setminus \Ccal^\prime \colon x_{H^\prime,c}=1  \text{ for some } H^\prime \in \consets(G) \}$.
Suppose that~$K\neq \emptyset$, otherwise we are trivially done.
For every~$c \in K$, we denote by~$H_c$ the subgraph in~$\consets_\cap(H)$ that is assigned color~$c$ by~$x$.
Let~$c^* \in \arg \min_{c \in K} |H\setminus H_c|$.
If~$|V(H)\setminus V(H_{c^*} )| \geq |\Ccal^\prime|$, then~$B(x)=0$.
In this case, the vector~$x^\prime$ is obtained from~$x$ by setting~$x^\prime_{H_c, c}=0$ for all~$c \in K$.
Note that~$x^\prime$ is an integer vector in~$\polytope(G,C)$ and satisfies~$f(x^\prime)=A(x^\prime)=A(x)=f(x)$.
If~$|H\setminus H_{c^*} | < |\Ccal^\prime|$, then at least one color in~$\Ccal^\prime$, say $c^\prime$, is not assigned by~$x$.
We now define~$x^\prime$ from~$x$ by removing color~$c^*$ of~$H_{c^*}$ and assigning color~$c^\prime$ to it, that is, $x^\prime= x-e(H_{c^*},c^*)+e(H_{c^*},c^\prime)$.
Since $A(x^\prime)-A(x)=B(x)-B(x^\prime)=|C^\prime|- |H \setminus H_{c^*}|$, we have again $f(x^\prime)=f(x)$. 
Note that, if~$x^\prime$ still uses a color in~$\Ccal\setminus \Ccal^\prime$, then we can repeat the previous process on it, a finite number of times, to get an integer vector in~$\polytope_k(G,C)$ that only assigns colors of~$\Ccal^\prime$.
In both cases, we have that $f(x')=f(x)$ and~$x^\prime$ only assigns colors from~$\Ccal^\prime$ to graphs in~$\consets_\cap(H)$.

Suppose now that~$|\Ccal^\prime|\leq |H|$.
Let~$\consets^\prime \subseteq \consets_\cap(H)$ be the subset of graphs used by~$x^\prime$ and~$\consets^{\prime\prime}=\{H^\prime \in \consets^\prime \colon \delta_{H^\prime} > 1\}$. 
Since each $H^\prime \in \consets^\prime$ is assigned to exactly one color in~$\Ccal^\prime$, we have $f(x^\prime)= A(x^\prime)= |\consets^\prime \setminus \consets^{\prime \prime}| + \sum_{H^\prime \in \consets^{\prime \prime} } \delta_{H^\prime}$.
If $\consets^{\prime\prime}=\emptyset$, then~$f(x^\prime)=|\consets^\prime|\leq |\Ccal^\prime|$, where the inequality comes from the fact that every graph in~$\consets^\prime$ is assigned by~$x^\prime$ to a different color in~$\Ccal^\prime$. 
Suppose now that~$\consets^{\prime\prime}\neq\emptyset$.
Since every graph in~$\consets^\prime$ intersects~$H$, it holds that~$|\consets^\prime \setminus \consets^{\prime\prime}|\leq \sum_{H^\prime \in \consets^\prime\setminus \consets^{\prime\prime} } |H\cap H^\prime|$.
Thus,  we obtain the following sequence of inequalities:
\begin{align*}
f(x^\prime) & = |\consets^\prime \setminus \consets^{\prime\prime}| + \sum_{H^\prime \in \consets^{\prime\prime}}(|\Ccal^\prime| -|H|+|H\cap H^\prime|)\\
 &  \leq \sum_{H^\prime \in \consets^\prime \setminus \consets^{\prime\prime}}| H \cap H^\prime | + \sum_{H^\prime \in \consets^{\prime\prime}}(|\Ccal^\prime| -|H|+|H\cap H^\prime|) \\
& \leq \sum_{H^\prime \in \consets^\prime} |H \cap H^\prime| + (|\Ccal^\prime|-|H|),
\end{align*}
where the last inequality is due to~$\consets^{\prime\prime} \neq\emptyset$ and~$|\Ccal^\prime|\leq |H|$.
Since the graphs in~$\consets^\prime$ are mutually disjoint and each of them intersects~$H$, it holds that~$\sum_{H^\prime \in \consets^\prime}|V(H)\cap V(H^\prime)| \leq |V(H)|$. 
Hence, $f(x^\prime)\leq |\Ccal^\prime|$ and therefore~$f(x)\leq |\Ccal^\prime|$, which proves the validity of~\eqref{eq:general}.

Finally, if $V(H)=\{v\}$, then~\eqref{eq:general} becomes~$\sum_{c \in \Ccal^\prime}\sum_{H^\prime : v\in V(H^\prime)} |\Ccal^\prime| x_{H^\prime,c } \leq |\Ccal^\prime|$. 
This is inequality~\eqref{ineq:vertices} multiplied by $|\Ccal^\prime|$ and so it is valid. 
\end{proof}

\begin{theorem}\label{thm:general-facet}
Let $|\Ccal| \geq 3$.
If~$|\Ccal^\prime|\leq |H|-1$ or $|H|=1$, then inequality~\eqref{eq:general} is valid and defines a facet of~$\polytope(G,C)$.
\end{theorem}
\begin{proof}
Due to space limitation, we omit the proof that inequality~\eqref{eq:general} is valid for ~$\polytope(G,C)$.
If $|H|=1$, then~\eqref{eq:general} is equivalent to~\eqref{ineq:vertices}.
If~$|\Ccal^{\prime}|=1$, then~\eqref{eq:general} reduces to~\eqref{eq:ineq1}. 
In both cases, it is facet-defining by Theorem~\ref{thm:facet}. 
Suppose from now on that~$|\Ccal^\prime| \geq 2$ and~$|\Ccal^\prime| \leq |H|-1$, and so~$|H| \geq 3$.

Let~$(\pi, \pi_0)$ be inequality~\eqref{eq:general} and let~$F$ be the face induced by~$(\pi, \pi_{0})$, that is, \(F=\{ x \in \polytope(G,C) \colon \pi x = \pi_0\}=\{x \in \polytope(G,C)  \colon A(x)+B(x)=|\Ccal^\prime|\}.\)
We shall prove that~$F$ is a facet by supposing that~$F \subseteq \{x \in \polytope(G,C) \colon \lambda x = \lambda_0 \}$ and showing that~$(\lambda,\lambda_0) = \theta (\pi,\pi_0)$, for some scalar~$\theta$. 
Actually, since~$F\neq \emptyset$, it is enough to prove~$\lambda=\theta \pi$, which implies $\lambda_0=\theta \pi_0$.
Let~$H^\prime \in \consets(G)$ and~$c^\prime \in \Ccal$. 

First we show the null entries in~$\lambda$.
There are two cases to be considered.
Since~$|\Ccal^\prime|\geq 2$, there is~$c^{\prime \prime} \in \Ccal^\prime \setminus \{c^\prime\}$. 
If~$H^\prime \notin \consets_\cap(H)$, then $e(H,c^{\prime\prime})$ and~$e(H,c^{\prime\prime})+e(H^\prime,c^{\prime})$ belong to~$F$.
As a consequence, we have~$\lambda_{H^\prime,c^\prime}=0$.
Suppose now that~$H^\prime \in \consets_\cap(H)$, $| H\setminus H^\prime| \geq |\Ccal^\prime|$, and~$c^\prime \notin \Ccal^\prime$.
Since~$|H \setminus H^\prime| \geq |\Ccal^\prime|$, we can define a vector~$x \in \polytope(G,C)$ that assigns each color in~$\Ccal^\prime$ to a different vertex of $H \setminus H^\prime$.
Hence, we have~$A(x)=\sum_{c \in \Ccal^\prime} \max(|\Ccal^\prime| -(|H|-1),1)= |\Ccal^\prime|$ and~$B(x)=0$. 
Similarly, the vector~$x^\prime=x+e(H^\prime, c^\prime) \in \polytope(G,C)$ satisfies $A(x^\prime)=|\Ccal^\prime|$ and~$B(x^\prime)=0$. 
Therefore, both~$x$ and~$x^\prime$ belong to~$F$ which implies~$\lambda_{H^\prime, c^\prime}=0$.

Now let that~$H^\prime \in \consets_\cap(H)$, $| V(H)\setminus V(H^\prime)| \geq |\Ccal^\prime|$, and~$c^\prime \in \Ccal^\prime$.
Let us define~$\theta=\lambda_{H^\prime, c^\prime}$.
We next show that~$\lambda_{H^{\prime\prime}, c^{\prime\prime}}=\theta$ for every~$H^{\prime\prime} \in \consets_{\cap}(H)$ with~$|H \setminus H^{\prime\prime}| \geq |\Ccal^\prime|$ and every~$c^{\prime\prime}\in \Ccal^\prime$.
Let us fix such~$H^{\prime\prime}$ and~$c^{\prime\prime}$.
Consider first the case in which~$H^\prime=\{u\}$ and~$H^{\prime\prime}=\{v\}$.
By transitivity, it suffices to analyze the following two subcases:
\begin{enumerate}
\item[$c^{\prime}=c^{\prime\prime}$:]
Let~$x\in \polytope(G,C)$ be the solution obtained by assigning~$c^\prime$ to~$u$ and each color in~$\Ccal^\prime \setminus \{c^\prime\}$ to a different vertex of~$H\setminus\{u,v\}$. 
This is possible because $|\Ccal^\prime| \leq |H|-1$ and~$c^\prime \in \Ccal^\prime$. 
Let us define~$x^\prime = x-e(u, c^\prime)+e(v,c^\prime) \in \polytope(G,C)$.
Observe that $A(x)=A(x^\prime)=|\Ccal^\prime|$ and~$B(x)=B(x^\prime)=0$. 
Therefore, we obtain~$x,x^\prime \in F$. 
It follows that~$\lambda(x-x^\prime)=0$ or still~$\theta=\lambda_{u, c^\prime}=\lambda_{v, c^\prime}$.

\item[$u=v:$]
Since $|H|\geq 3$, there is a vertex~$z$ of~$H$ such that~$K:=H-z$ is connected.
As~$|\Ccal | \geq 3$, there must exist $\hat c \in \Ccal \setminus \{c^\prime, c^{\prime\prime}\}$. 
Let~$x=e(z, c^\prime)+e(K, \hat c)$ and~$x^\prime =e(z, c^{\prime\prime})+e(K, \hat c)$ and observe that these vectors belong to~$\polytope(G,C)$.
Since $c^\prime,c^{\prime\prime}\in \Ccal^\prime$ and~$|V(H) \setminus V(K)|=1$, we have~$A(x)+B(x)=A(x^\prime)+B(x^\prime)=|\Ccal^\prime|$ regardless whether~$\hat c$ belongs to~$\Ccal^\prime$ or not. 
Again, we obtain~$x,x^\prime \in F$ and so $\lambda_{z, c^\prime}=\lambda_{z, c^{\prime\prime}}$. 
By the previous item,  we conclude that~$\theta=\lambda_{u, c^\prime}=\lambda_{z, c^\prime}=\lambda_{z, c^{\prime\prime}}=\lambda_{u, c^{\prime\prime}}$. 
\end{enumerate}

Now, let~$x \in \polytope(G,C)$ be the solution obtained by assigning~$c^{\prime\prime}$ to~$H^{\prime\prime}$ and each color in~$\Ccal^\prime \setminus \{c^{\prime\prime}\}$ to a different vertex in~$H \setminus H^{\prime\prime}$.
Since~$|H\setminus H^{\prime\prime}| \geq |\Ccal^\prime|$, this solution is possible and still leaves a vertex~$z \in H \setminus H^{\prime\prime}$ uncolored.
We define~$x^\prime=x-e(H^{\prime\prime}, c^{\prime\prime})+e(z, c^{\prime\prime})$.
Thus, we obtain 
\begin{align*}
	A(x)& =\sum_{c \in \Ccal^\prime \setminus\{c^{\prime\prime}\}} \max(|\Ccal^\prime|-(|H|-1),1)+ \max(|\Ccal^\prime| - |H\setminus H^{\prime\prime}|,1)=|\Ccal^\prime|,\\
	A(x^\prime) &  = \sum_{c \in \Ccal^\prime} \max( |\Ccal^\prime| - (|H|-1), 1 ) = |\Ccal^\prime| \\
\end{align*}
and~$B(x)=B(x^\prime)=0$. 
Therefore, we conclude that~$x$ and~$x^\prime$ belong to~$F$, and so~$\lambda_{H^{\prime\prime}, c^{\prime\prime}}=\lambda_{z, c^{\prime\prime}}=\theta$. 

Finally, suppose that~$H^\prime \in \consets_\cap(H)$,  $|H\setminus H^\prime|< |\Ccal^\prime|$, and~$c^\prime \in \Ccal$.
We shall prove that~$\lambda_{H^\prime, c^\prime}=(|\Ccal^\prime| - |H\setminus H^\prime|)\theta$.
Let~$x\in \polytope(G,C)$ be obtained by assigning~$c^\prime$ to~$H^\prime$ and~$|H\setminus H^\prime|$ colors from~$C^\prime \setminus\{c^\prime\}$ to different vertices of~$H\setminus H^\prime$. 
Regardless whether~$c^\prime$ belongs to~$\Ccal^\prime$ or not, it holds that~$A(x)+B(x)=|H\setminus H^\prime| + (|\Ccal^\prime|-|H\setminus H^\prime|) =|\Ccal^\prime|$ and thus~$x\in F$.

Let~$S$ be the set of colors in $\Ccal^\prime \cup \{c^\prime\}$ assigned by~$x$. 
We define~$x^\prime$ from~$x$ by removing~$c^\prime$ and assigning a subset~$S^\prime \subseteq \Ccal^\prime \setminus (S\setminus \{c^\prime\} )$ of colors, where $|S^\prime|=|\Ccal^\prime|-|V(H)\setminus V(H^\prime)|$, to distinct vertices in~$H \cap H^\prime$. 
Note that there is such an $S^\prime$ since~$|S \setminus \{c^\prime\}| \leq |H\setminus H^\prime|$.
Besides, $x^\prime$ is a possible solution because \(|H \cap H^\prime| = |H| - |H \setminus H^\prime| > |\Ccal^\prime| -| H \setminus H^\prime |\) since~$|H| > |\Ccal^\prime| $.
This implies~$x^\prime \in \polytope(G,C)$.
Moreover, as $A(x^\prime)=|\Ccal^\prime|$ and~$B(x^\prime)=0$, we have $x\in F$. 
Therefore, it follows that~$0=\lambda(x-x^\prime)=\lambda_{H^\prime, c^\prime}-|S^\prime| \theta$.
This completes the proof  that~$F$ is a facet of~$\polytope(G,C)$. 
\end{proof}

\section{Computational Experiments}
\label{sec:experiments}
We restricted our computational experiments to path instances. In this case, the number of connected subgraphs is polynomial so that a column generation approach for \eqref{fob}-\eqref{ineq:integer} is not mandatory. Besides, the CR on paths is closely related to the Connected Assignment Problem in Arrays (CAPA), a problem recently introduced in~\cite{CamSoaMacLim2020} to model an important application in 4G mobile networks.  

The experiments were run on an AMD FX-4300 Quad-Core (3.80 GHz) with 6 GB DDR3 (1066 MHz) and Ubuntu 18.04 LTS of 64 bits.  We evaluated the reduction in the integrality gaps provided by inequalities~\eqref{eq:ineq1}. We used the CPLEX 12.6.1 optimizer to solve the LP programs. The inequalities were added with the CPLEX addCuts method.

For each tested instance, we calculated the percentage gaps $\% G_i=(LP_i-OPT)/OPT$, $i \in \{0,1\}$, where $OPT$ is the optimal value of the integer program, $LP_0$ is the optimal value of the linear program given by constraints~\eqref{eq:ineq1}, for all $c\in \Ccal$ and all $H\in \consets(G)$ with $|H|=1$ or $H=V(G)$ (which are exactly \eqref{ineq:vertices} and a lifting of \eqref{ineq:colors}), and $LP_1$ is the optimal value of the linear program given by constraints~\eqref{eq:ineq1}, for all $c\in \Ccal$ and all $H\in \consets(G)$. In addition, we calculated $\% GR=(\% G_0-\% G_1)/\% G_{0}$, which represents the percentage gap reduction.

\subsection{Computational results for CR on paths}
Three hundred test instances were randomly generated, 20 for each pair $(n,k=\alpha\lceil n/4\rceil)$, with $n\in\{20,25,30,35,40\}$ and $\alpha\in \{1,2,3\}$. 
The initial color of each vertex was uniformly chosen from the range $[1,k]$; unit weights were assumed.

Figure~\ref{fig:CR-GR} presents the percentage gap reduction per instance. In the $x$-axis the instances are ordered according to the value of $n$ and then $k$. 
Points over the bottom horizontal line correspond to cases where constraints~\eqref{eq:ineq1} were not able to improve the gap, whereas points over the top horizontal line represent instances where the gap was closed after adding these constraints. It is worth mentioning that the solution giving $LP_0$ was not integer for every instance (although $LP_0=OPT$ in some cases).
Overall, $\% G_1$ improved $\% G_0$ in 13.67\% of the instances leading to an average $\% GR$
(for the 300 instances) of 12.73\%. More importantly, the solution yielding $LP_1$ was integer in 76.33\% of the cases. It means that the added inequalities were fundamental to prove optimality without branching.


Figure~\ref{fig:CR-GR2} plots $\% GR$ versus $\% G_0$. First, we can observe that $\% G_0$ is already small (usually less than 5\%), which indicates the strength of the basic formulation. Even when $\% G_0$ is small, the additional inequalities could improve the upper bound, either closing the gap (as show the points in the top horizontal line) or reducing it (as show the points in the intermediate horizontal region). 

\begin{figure}
  \centering
  \subfigure[Percentage gap reduction per instance. \label{fig:CR-GR}]{\includegraphics[scale=0.7, angle=90]{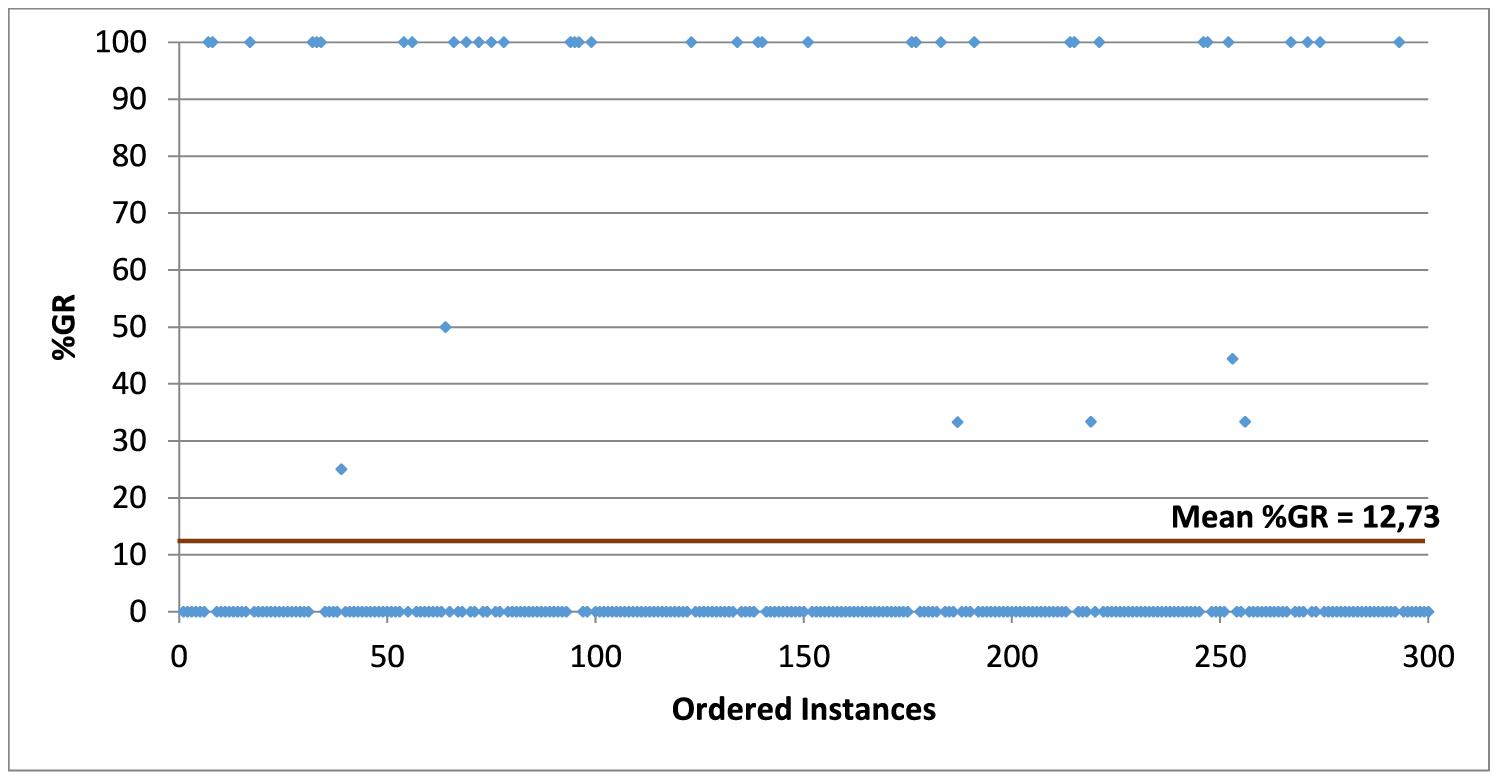}}
  \hfil
  \subfigure[Percentage gap reduction versus initial percentage gap. \label{fig:CR-GR2}]{\includegraphics[scale=0.7, angle=90]{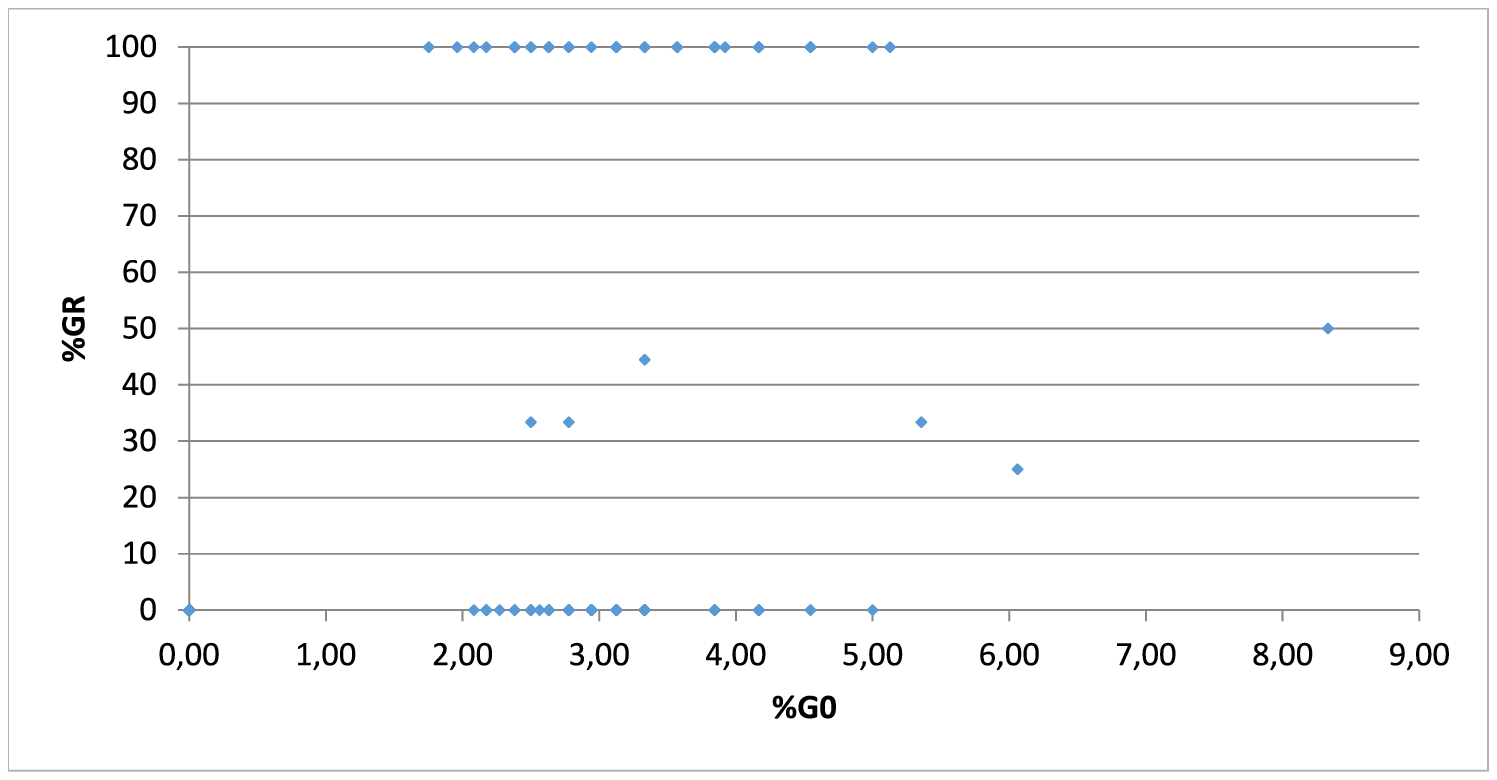}}
  \caption{CR- Percentage gap reduction.}
\end{figure}

\begin{figure}
\centering
\includegraphics[scale=0.6]{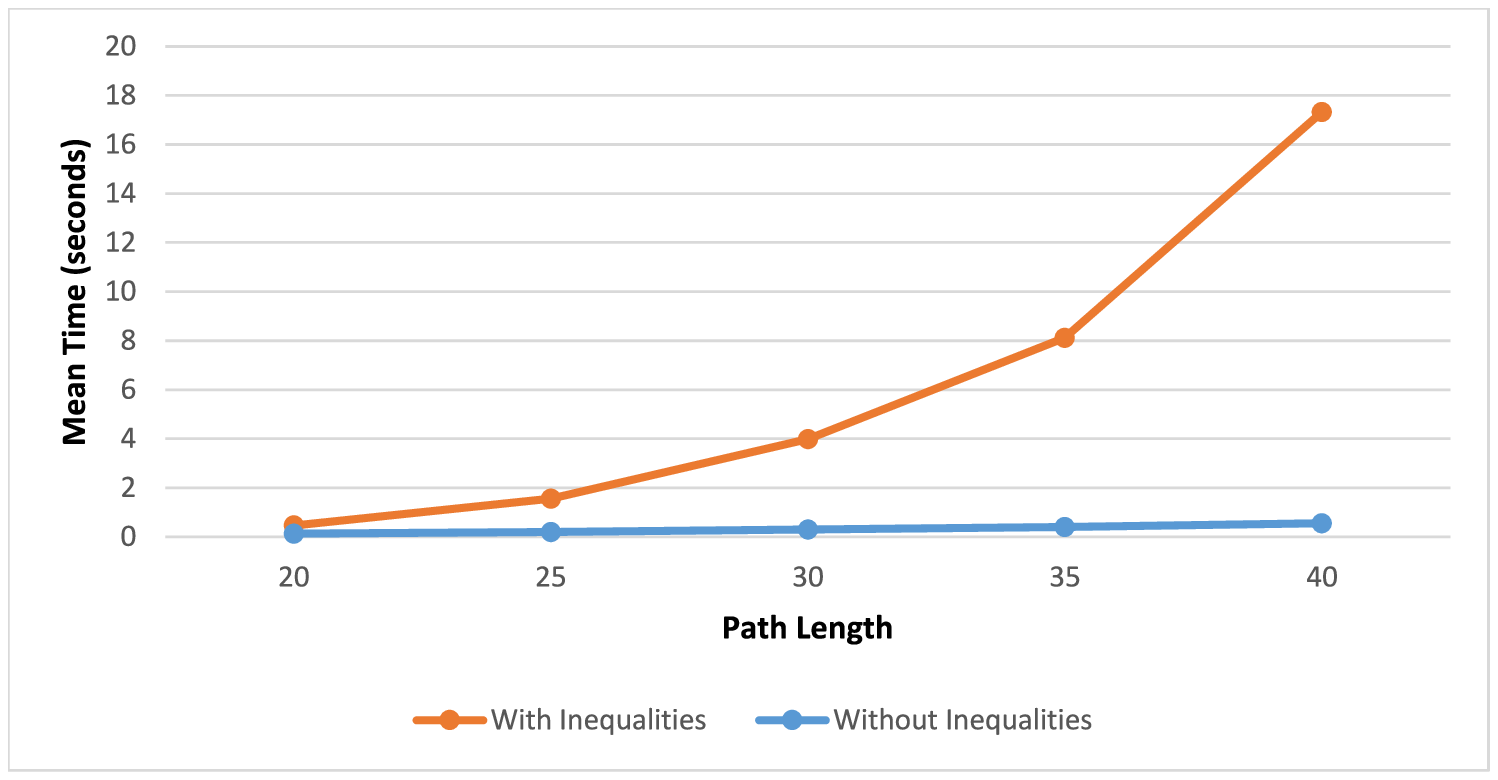}
\caption{CR - Formulation Processing Time.}
\label{fig:CR-T}
\end{figure}

This gain in the gap or proof of optimality had a cost. Figure~\ref{fig:CR-T} shows that the running time grows exponentially with $n$ when the inequalities are added. It suggests the design of a smart cutting plane process to circumvent this drawback.

\subsection{Computational results for CAPA}
Given an array with $n$ positions, a set with $k$ symbols and a gain matrix $[\rho_{ij}]_{k\times n}$, the CAPA problem consists in assigning at most one symbol $i$ to each position $j$ so as to get a gain $\rho_{ij}$. The sum of the gains is to be maximized under the constraint that repeated symbols must appear in consecutive positions of the array. This problem could be seen as a variant of the CR problem on paths where we are given a $n$-vertex uncolored path and want to find a convex recoloring. Now, the weight of the recoloring depends on the color assigned to each vertex. In other words, the solutions of CAPA can be defined by constraints~\eqref{ineq:vertices}-\eqref{ineq:integer} and the coefficient of $x_{H,c}$ in the objective function~\eqref{fob} is $\sum_{j\in H}\rho_{cj}$ (this formulation for CAPA was introduced in \cite{LimMacCav2016}). Therefore, the polyhedral results obtained here can be directly applied to CAPA.

Again, 300 instances (with the sizes described as before) were generated. Now, the gain matrix $\rho$ is obtained by the same simulator used in \cite{LimMacCav2016}. It models and simulates an LTE-like cellular system. Matrix $\rho$ corresponds to achievable data rates according to channel quality indicators that capture the most relevant propagation effects of cellular systems. 

Figures~\ref{fig:CAPA-GR}, \ref{fig:CAPA-GR2} and \ref{fig:CAPA-T} show similar graphics for CAPA. The average $\% GR$
was of 6.73\%. $\% G_1$ improved $\% G_0$ in only 8.33\% of the instances. However, $\% G_0$ was already very small (usually less than 1\%). Now there were more instances where the initial gap was reduced but not closed. The solution yielding $LP_1$ was integer in 86\% of the 300 CAPA instances.
 Again, an exponential increase with~$n$ could be observed in the running time when the inequalities are added to the formulation.
  In the application context, it is worth mentioning that the used instance sizes are compatible with real scenarios, but the processing times are already prohibitive. 



\begin{figure}
  \centering
  \subfigure[Percentage gap reduction per instance. \label{fig:CAPA-GR}]{\includegraphics[scale=0.7, angle=90]{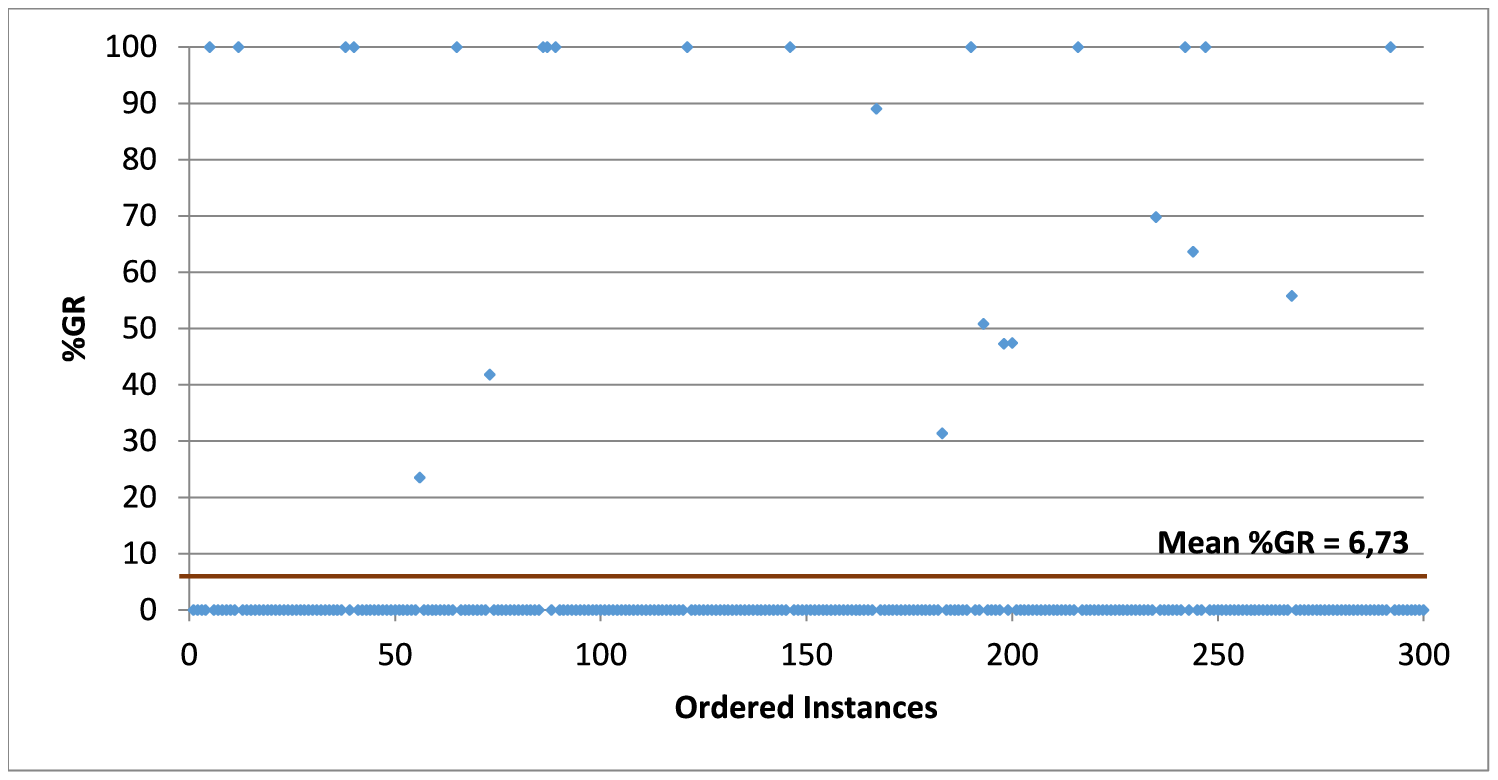}}
  \hfil
  \subfigure[Percentage gap reduction versus initial percentage gap. \label{fig:CAPA-GR2}]{\includegraphics[scale=0.7, angle=90]{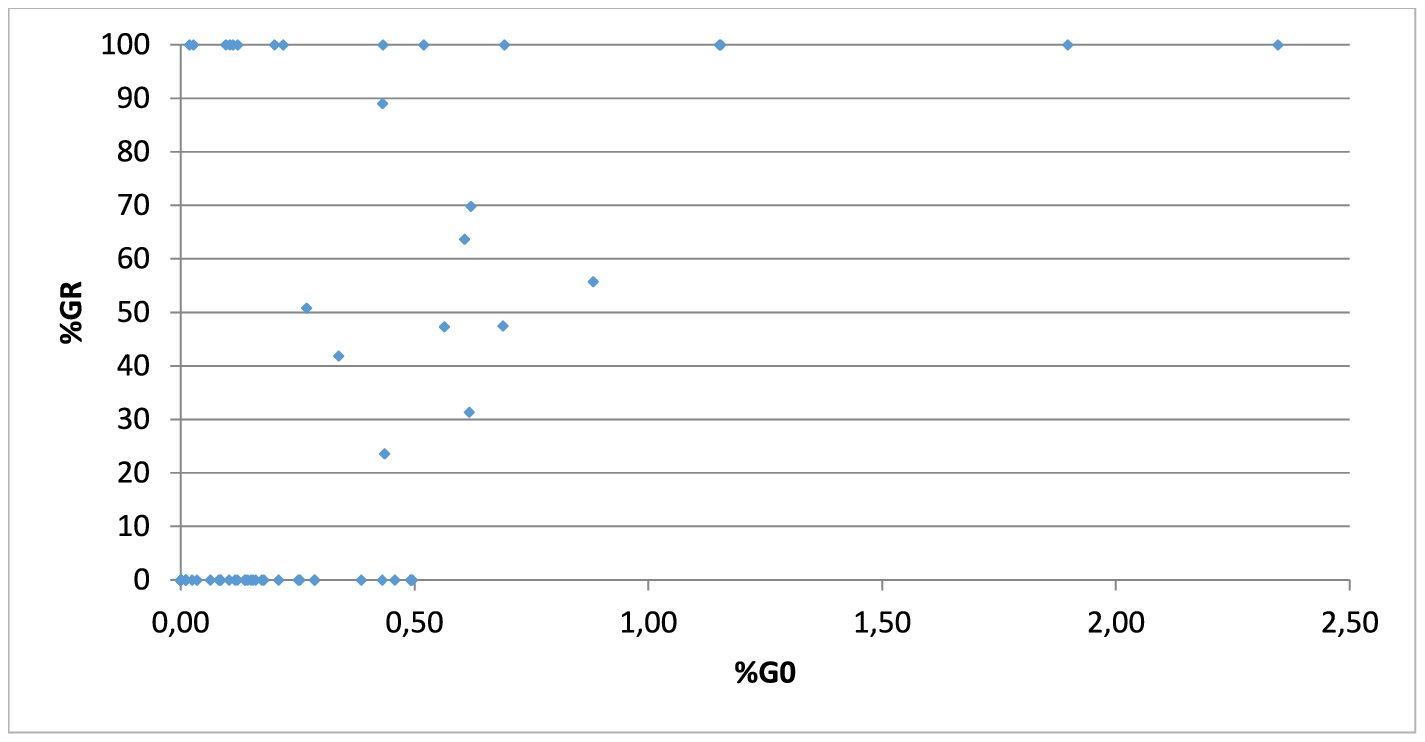}}
  \caption{CAPA - Percentage gap reduction.}
\end{figure}

\begin{figure}
\centering
\includegraphics[scale=0.7]{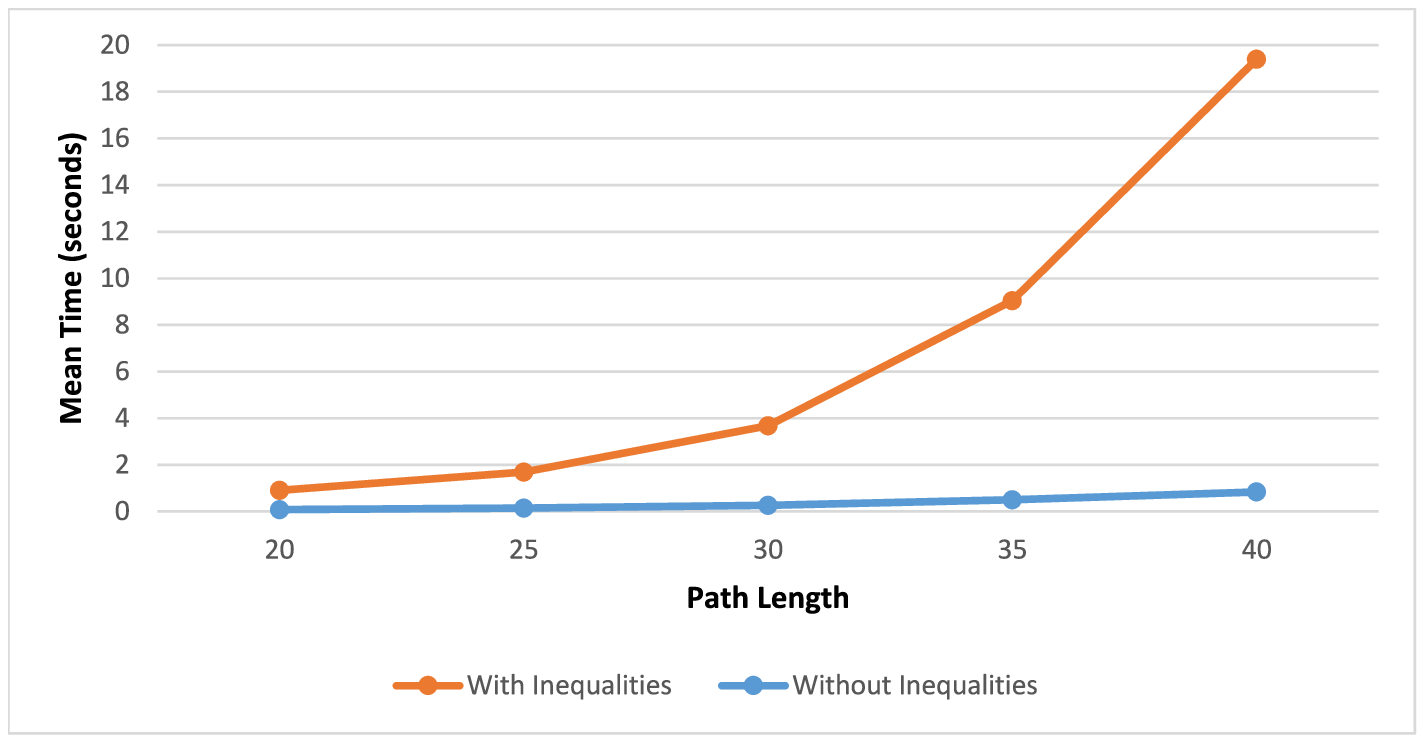}
\caption{CAPA - Formulation Processing Time.}
\label{fig:CAPA-T}
\end{figure}
\section{Conclusion and further research}

We presented a polytope for \CR\ based on the idea connected subgraphs and studied its facial structure.
We devised classes of facet-defining inequalities and implemented cutting plane algorithms for \CR\ on paths and for an application on mobile networks.
The computational experiments revealed that the linear relaxation of the formulation already gives a small gap (usually less than 5\%), which indicates the strength of this formulation.
Furthermore, they showed that the added inequalities were fundamental to prove optimality without branching.

Motivated by applications in the study of phylogenetic trees~\cite{MorSni08} and by the promising computational results of column-generation approaches presented in~\cite{ChoErdKimShi17,Mou17}, we aim to design efficient separation routines for the proposed inequalities, implement a branch-cut-and-price algorithm for \CR\ on trees and run experiments to evaluate its performance.

\bibliographystyle{abbrv}%


\end{document}